\documentclass[journal,10pt]{IEEEtran}
\usepackage{amsmath,amssymb,amsfonts}
\usepackage{graphicx}
\usepackage{booktabs}
\usepackage{cite}
\usepackage{url}
\usepackage{xcolor}
\usepackage{mathtools}
\usepackage{bm}
\usepackage{microtype}
\usepackage{balance}

\newcommand{\E}{\mathbb{E}}
\newcommand{\Var}{\operatorname{Var}}
\newcommand{\CN}{\mathcal{CN}}
\newcommand{\GammaregL}{P}
\newcommand{\GammaregU}{Q}
\newcommand{\bep}{\overline P_b}
\newtheorem{theorem}{Theorem}
\newtheorem{corollary}{Corollary}

\begin{document}
% \title{Fluid Antenna Multiple Access for Two-Level Noise Modulation:\\
% Exact Port Selection, Common-State Coupling, and Load Scaling}

\title{Fluid Antenna Multiple Access for Noise Modulation}

\author{Hadi Zayyani, Felipe A. P. de Figueiredo
\thanks{Authors are with INATEL, Brazil.}}

\maketitle

\begin{abstract}
Noise modulation (NoiseMod) encodes information in the variance of a noise-like waveform. Its noncoherent structure is attractive for low-complexity links, but co-channel NoiseMod users are challenging because their waveforms enter the same variance statistic used for detection. We study fluid antenna multiple access (FAMA) for \emph{two-level} NoiseMod when every interferer transmits an independent random bit and therefore contributes one of two nonzero interference variances. Exact finite-sample maximum-likelihood energy detection is adopted from prior NoiseMod/TNM work and is not claimed as a new detector. We derive the exact single-port moment-generating function and moments of the random-bit aggregate interference, including heterogeneous user powers, and prove that the port minimizing the exact conditional bit-error probability (BEP) is $\arg\max_k H_k/(\sigma_w^2+J_k)$, where $H_k$ is the desired-link power and $J_k$ is the instantaneous bit-state-aware aggregate interference variance. For independent fading ports, the common interferer bit vector couples the complete per-port decision metrics even when the fading itself is independent; conditioning on the number of high-state interferers yields an exact binomial-mixture order-statistic representation and avoids the generally invalid unconditional rule $F_{Z_{\max}}=F_Z^{N_p}$. For spatially correlated operation, an integer-$\mu$ $\kappa$-$\mu$ benchmark is constructed by applying full pairwise Jakes correlation to the scattered cluster fields, and channel-averaged BEP is evaluated by conditional Monte Carlo. A standard-error-aware load study shows a large increase in admissible co-channel users as useful spatial degrees of freedom grow. Finally, a finite-sample variance-domain sensing diagnostic quantifies the oracle gap under the deliberately favorable assumption that channels and interferer states remain frozen throughout sensing and the following data bit. Its large sample overhead is therefore interpreted as a lower-bound warning on fast-FAMA sensing difficulty, not as evidence of an immediately deployable acquisition protocol.
\end{abstract}

\begin{IEEEkeywords}
Fluid antenna multiple access (FAMA), fluid antenna system (FAS), noise modulation, co-channel interference, $\kappa$-$\mu$ fading, port selection, bit error probability.
\end{IEEEkeywords}

\section{Introduction}
Noise modulation (NoiseMod) conveys information through the variance of a random waveform rather than its instantaneous phase or amplitude. Basar introduced the general NoiseMod framework and analyzed variance-based detection using a Gaussian/central-limit approximation for the sample-energy statistic \cite{basar_noisemod}. Exact finite-sample maximum-likelihood (ML) energy detection was subsequently derived for thermal-noise modulation \cite{alshawaqfeh_tnm}, and recent work has emphasized that NoiseMod's apparent per-sample-SNR advantages must be interpreted together with sample count, bit energy, rate, threshold calibration, and fading sensitivity \cite{figueiredo_practical}. Thus, exact single-link detection is already prior art; the question here is how that detector behaves in a genuinely multi-user, spatially reconfigurable receiver.

FAMA exploits spatial variation over a compact fluid-antenna aperture to find a port with a favorable desired-to-interference condition, enabling open-loop multiple access without transmit precoding \cite{wong_fama}. Noise-aware slow-FAMA analysis and later partial-observation work already use SINR as a performance or port-selection quantity \cite{yang_slow_noisy,eskandari_cgan}; accordingly, we do not claim that a generic SINR-form ranking is itself new. The new point proved here is that the exact finite-sample NoiseMod BEP induces the particular instantaneous statistic $H_k/(\sigma_w^2+J_k)$ when $J_k$ is the current random-bit variance field. Spatial modeling matters: the simple reference-port model used in early FAS/FAMA analysis can misrepresent the full Jakes correlation among closely spaced ports, motivating full pairwise or otherwise more accurate correlation models \cite{khammassi}. Recent studies have independently combined FAS with NoiseMod \cite{zayyani_fas_noisemod} and with on--off digital noise (OODN) modulation over $\kappa$-$\mu$ fading \cite{araujo_oodn_fas}. Co-channel OODN interference has also been analyzed through a finite mixture over interferer activity states, including an admission-control interpretation \cite{araujo_oodn_interference}. Multi-user communication in the broader noise domain is likewise not new: noise-domain NOMA has multiplexed users by assigning information to distinct statistical dimensions such as mean and variance \cite{yapici_ndnoma}. That construction is different from the present same-dimension co-channel problem, in which all users modulate variance and become mutual interference. These works sharply narrow what can legitimately be claimed as new: neither exact energy detection, FAS diversity, $\kappa$-$\mu$ fading, nor multi-user noise-domain signaling is novel in isolation.

The unresolved interaction addressed here is different. In two-level NoiseMod both symbol states have nonzero variance, so every co-channel user is always present but its contribution changes with its bit. The same interferer bit vector is seen at every fluid-antenna port, while the corresponding channel powers vary spatially. Consequently, the interference field is bit-state dependent, the optimal port metric must be tied to the exact NoiseMod decision problem, and even an ``independent-port'' fading benchmark retains cross-port dependence through the common transmitted bits.

A second issue is practicality. Contemporary fast-FAMA work explicitly identifies the common genie-aided premise, probing all ports every symbol and knowing the desired/interference split, as unrealistic because of switching, pilot, and reconstruction overhead \cite{waqar_fast_fama}. Related work argues that selection-oriented channel reconstruction, rather than global NMSE alone, is the relevant sensing problem and highlights the tradeoff between multi-port sensing overhead and port-selection gain \cite{elganimi_channel_est}. We therefore use full per-port knowledge only as an \emph{oracle benchmark}. A deliberately favorable finite-sample variance-domain sensing experiment is included only to expose the oracle-to-sensing gap; because it freezes interferer states across the sensing interval and following data bit, it must not be interpreted as a complete symbol-by-symbol acquisition solution.

The contributions are as follows.
\begin{itemize}
\item For two-level equiprobable NoiseMod interferers, we derive the exact single-port MGF, mean, and variance of the aggregate interference $J=\sum_iP_{i,B_i}G_i$, including unequal average user powers. This makes explicit why random unequal state scalings do not, in general, collapse to one common-scale $\kappa$-$\mu$ variate.
\item We prove that, for the exact finite-$N_s$ ML detector, the conditional BEP is strictly decreasing in the variance ratio $V_1/V_0$. Hence, the exact BEP-minimizing port is $k^\star=\arg\max_k H_k/(\sigma_w^2+J_k)$, not merely an interference-limited SIR heuristic.
\item For iid fading across ports, we derive an exact order-statistic reduction conditioned on the common interferer bit state. For identical interferers it reduces to a binomial mixture over the number of high-state interferers. This identifies and quantifies the error in the naive unconditional $F_Z^{N_p}$ construction.
\item We introduce a multiple-access load metric, $N_I^{\max}$ at a target BEP, and study its dependence on aperture and port count under a full-pairwise Jakes-correlated integer-$\mu$ $\kappa$-$\mu$ benchmark, together with SIR-FAMA, max-$H$ FAS, min-interference, and fixed-port baselines and a near--far stress test.
\item We construct a finite-sample partial-probing and selected-port calibration diagnostic under a deliberately favorable frozen-state assumption. Its overhead/performance tradeoff quantifies an oracle-to-sensing gap and shows that direct probing is sample-inefficient even before fast interferer-state variation is introduced.
\end{itemize}

\section{Signal and Spatial Channel Model}
\label{sec:model}
\subsection{Two-Level NoiseMod With Co-Channel Users}
Consider one desired transmitter and $N_I$ co-channel NoiseMod interferers. User $u$ transmits $N_s$ independent complex Gaussian samples per bit,
\begin{equation}
 X_{u,n}\mid B_u=b\sim\CN(0,P_{u,b}),\qquad b\in\{0,1\},
\end{equation}
with equiprobable bits and a common variance ratio $\alpha>1$. We parameterize the two nonzero levels by the user's average transmitted variance $\bar P_u$,
\begin{equation}
 P_{u,0}=\frac{2\bar P_u}{1+\alpha},\qquad P_{u,1}=\alpha P_{u,0},
\label{eq:levels}
\end{equation}
so $\E[P_{u,B_u}]=\bar P_u$.

At fluid-antenna port $k$ the received sample for desired bit $b$ is
\begin{equation}
 Y_{k,n}=h_kX_{D,n}+\sum_{i=1}^{N_I}g_{i,k}X_{i,n}+W_{k,n},
\label{eq:rx}
\end{equation}
where $W_{k,n}\sim\CN(0,\sigma_w^2)$, $H_k=|h_k|^2$, and $G_{i,k}=|g_{i,k}|^2$. Desired and interfering users have independent channel processes. The $N_I$ interferer bits are independent and equiprobable, and are held constant over the $N_s$ samples of one desired bit. This is a synchronous block assumption; asynchronous state changes inside the observation window are outside the present model.

Conditioned on channels and interferer bits, the sum in \eqref{eq:rx} is exactly Gaussian because every transmitted waveform is Gaussian. Define
\begin{align}
 J_k&=\sum_{i=1}^{N_I}P_{i,B_i}G_{i,k},\label{eq:J}\\
 C_k&=\sigma_w^2+J_k,\label{eq:C}\\
 V_{b,k}&=C_k+H_kP_{D,b}.\label{eq:Vb}
\end{align}
Then $Y_{k,n}|b,\bm B,\bm H,\bm G\sim\CN(0,V_{b,k})$. Importantly, the same transmitted bit vector $\bm B=(B_1,\ldots,B_{N_I})$ enters \eqref{eq:J} at every port. This common state will matter in Section~\ref{sec:iid} even when fading across ports is independent.

\subsection{Jakes-Correlated Integer-$\mu$ $\kappa$-$\mu$ Benchmark}
For spatially correlated results we use an explicit joint benchmark whose single-port marginal is $\kappa$-$\mu$ with integer $\mu$. An $N_p$-port linear aperture spans $W\lambda$, with full pairwise scattered-field correlation
\begin{equation}
 [\bm R]_{k\ell}=J_0\!\left(\frac{2\pi(k-\ell)W}{N_p-1}\right),\qquad N_p>1,
\label{eq:jakes}
\end{equation}
and $\bm R=[1]$ for $N_p=1$. For a link with mean power $\Omega$, define
\begin{equation}
 v=\frac{\Omega}{\mu(1+\kappa)},\qquad d=\sqrt{v\kappa}.
\end{equation}
For cluster $c$, draw $\bm Z_c\sim\CN(\bm0,v\bm R)$ independently across $c$, and set
\begin{equation}
 H_k=\sum_{c=1}^{\mu}|d+Z_{c,k}|^2.
\label{eq:kmu_joint}
\end{equation}
The same construction is used independently for each $G_{i,k}$. Marginally, \eqref{eq:kmu_joint} is the standard noncentral-chi-square construction of $\kappa$-$\mu$ power with $\E[H_k]=\Omega$. We intentionally restrict $\mu$ to positive integers and use a common dominant-component phase across the compact aperture. The construction is therefore a transparent benchmark joint model, not a claim of a universal correlated $\kappa$-$\mu$ process. Its use of full pairwise Jakes correlation targets the correlation-model concern identified in \cite{khammassi}.

\begin{figure*}[t]
\centering
\includegraphics[width=0.94\textwidth]{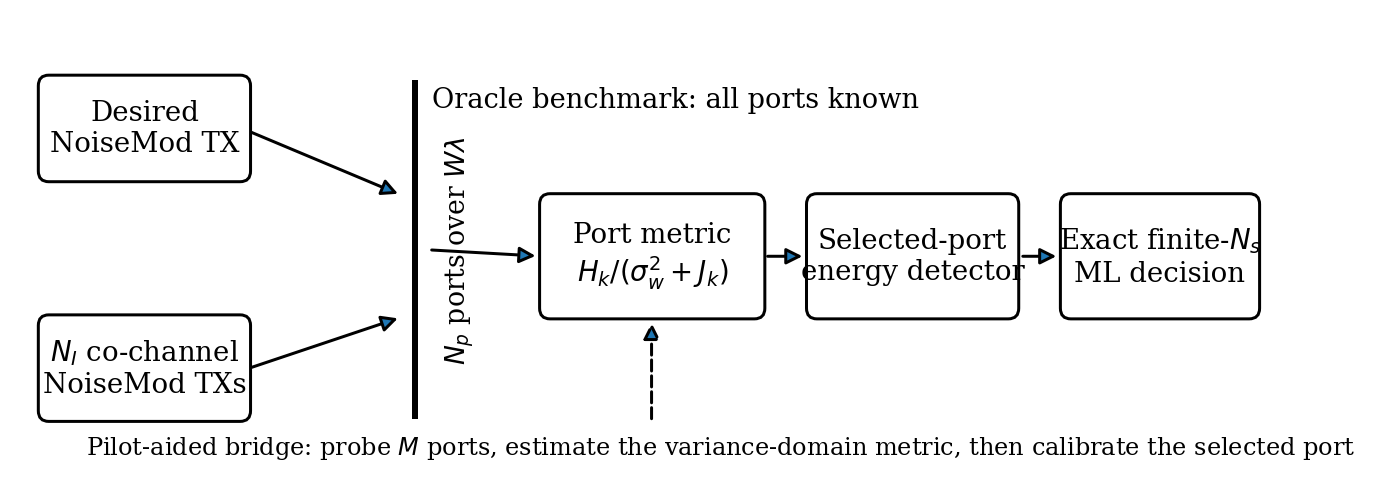}
\caption{System and evaluation architecture. The noise-aware metric is the exact finite-$N_s$ BEP-optimal oracle rule proved in Section~\ref{sec:selection}. The pilot-aided path is a finite-sample sensing diagnostic under a frozen-state assumption; it is not assumed to solve symbol-by-symbol fast-FAMA reconstruction.}
\label{fig:system}
\end{figure*}

\section{Exact Finite-Sample Detection}
\label{sec:detector}
Exact finite-sample likelihood-based energy detection is established prior art for noise-level/thermal-noise modulation \cite{alshawaqfeh_tnm,figueiredo_practical}; we summarize the complex-baseband Gamma specialization implied by our $\CN(0,V_b)$ model because it is the foundation of the port-selection theorem. The selected-port statistic is
\begin{equation}
 T=\frac{1}{N_s}\sum_{n=1}^{N_s}|Y_n|^2.
\end{equation}
Conditioned on bit $b$ and state $(H,J)$,
\begin{equation}
 \frac{N_sT}{V_b}\sim\operatorname{Gamma}(N_s,1).
\label{eq:gamma_stat}
\end{equation}
For equal priors and $0<V_0<V_1$, equating the two Gamma likelihoods gives the exact threshold
\begin{equation}
 \eta_{\rm ML}=\frac{V_0V_1}{V_1-V_0}\ln\!\frac{V_1}{V_0}.
\label{eq:eta}
\end{equation}
Let $\GammaregL(a,x)$ and $\GammaregU(a,x)$ denote the regularized lower and upper incomplete Gamma functions. The exact conditional BEP is
\begin{equation}
 p_e(V_0,V_1)=\frac12\left[
 \GammaregU\!\left(N_s,\frac{N_s\eta_{\rm ML}}{V_0}\right)
 +\GammaregL\!\left(N_s,\frac{N_s\eta_{\rm ML}}{V_1}\right)
 \right].
\label{eq:exact_bep}
\end{equation}
Unlike a CLT approximation, \eqref{eq:exact_bep} is exact under the iid complex-Gaussian sample model in Section~\ref{sec:model}. It is scale-invariant: $p_e(cV_0,cV_1)=p_e(V_0,V_1)$ for $c>0$.

\section{Bit-State-Aware Interference Statistics}
\label{sec:interference}
At a single port let $G$ denote one normalized $\kappa$-$\mu$ power from \eqref{eq:kmu_joint}. The MGF under the integer-cluster construction is
\begin{equation}
 M_G(s)=(1-vs)^{-\mu}\exp\!\left(\frac{\mu\kappa v s}{1-vs}\right),\qquad s<1/v.
\label{eq:mgfG}
\end{equation}
For independent users and equiprobable bits, the exact marginal MGF of \eqref{eq:J} is therefore
\begin{equation}
 M_J(s)=\prod_{i=1}^{N_I}\frac{M_G(P_{i,0}s)+M_G(P_{i,1}s)}{2},
 \quad s<\frac{1}{v\max_i P_{i,1}}.
\label{eq:mgfJ}
\end{equation}
Equation~\eqref{eq:mgfJ} remains valid for heterogeneous received average powers $\bar P_i$; the stated positive-$s$ domain is the intersection of the component-MGF domains (and is automatically satisfied for Laplace-transform checks with $s<0$). It also makes clear why a single common-scale $\kappa$-$\mu$ aggregation is generally invalid: the random state multiplies $G_i$ by either $P_{i,0}$ or $P_{i,1}$ before summation.

For later validation, the $\kappa$-$\mu$ power variance is
\begin{equation}
 \sigma_G^2=\Omega^2\frac{1+2\kappa}{\mu(1+\kappa)^2}.
\end{equation}
With $q_{2,i}=\tfrac12(P_{i,0}^2+P_{i,1}^2)$ and $\E[P_{i,B_i}]=\bar P_i$,
\begin{align}
 \E[J]&=\Omega\sum_i\bar P_i,\label{eq:Jmean}\\
 \Var(J)&=\sum_i\left[q_{2,i}(\Omega^2+\sigma_G^2)-\bar P_i^2\Omega^2\right].
\label{eq:Jvar}
\end{align}
The closure of noncentral chi-square variables still applies \emph{conditioned on common scales}; what fails is replacing the unconditional random mixture of differently scaled terms by one same-family random variable.

\section{Exact BEP-Optimal Port Selection}
\label{sec:selection}
Define $\rho=V_1/V_0>1$. By scale invariance, \eqref{eq:exact_bep} can be written as $p_e(\rho)$. The next result links the exact detector directly to a FAMA metric.

\begin{theorem}[Exact NoiseMod-FAMA selection rule]
\label{thm:selection}
For equiprobable two-level NoiseMod with $P_{D,1}>P_{D,0}>0$ and the finite-$N_s$ ML detector \eqref{eq:eta}, the conditional BEP is strictly decreasing in $\rho=V_1/V_0$. Consequently, among simultaneously available ports, the port minimizing exact conditional BEP is
\begin{equation}
 k^\star=\arg\max_k Z_k,\qquad Z_k\triangleq\frac{H_k}{\sigma_w^2+J_k}.
\label{eq:optimal_port}
\end{equation}
\end{theorem}

\begin{IEEEproof}
Let $f_N(x)=x^{N_s-1}e^{-x}/\Gamma(N_s)$ be the unit-scale Gamma density and define
\begin{equation}
 a(\rho)=\frac{N_s\ln\rho}{\rho-1},\qquad z(\rho)=\rho a(\rho).
\end{equation}
From \eqref{eq:eta}--\eqref{eq:exact_bep},
\begin{equation}
 p_e(\rho)=\tfrac12\left[\GammaregU(N_s,z)+\GammaregL(N_s,a)\right].
\end{equation}
Because $f_N(a)=\rho f_N(z)$, direct differentiation yields
\begin{equation}
 \frac{dp_e}{d\rho}=-\frac{z f_N(z)}{2\rho}<0.
\label{eq:derivative}
\end{equation}
At port $k$, with $C_k=\sigma_w^2+J_k$ and $Z_k=H_k/C_k$,
\begin{equation}
 \rho_k=\frac{1+P_{D,1}Z_k}{1+P_{D,0}Z_k},
\end{equation}
whose derivative with respect to $Z_k$ is $(P_{D,1}-P_{D,0})/(1+P_{D,0}Z_k)^2>0$. Thus maximizing $Z_k$ maximizes $\rho_k$ and, by \eqref{eq:derivative}, minimizes the exact BEP.
\end{IEEEproof}

\begin{corollary}[Interference-limited limit]
If $J_k\gg\sigma_w^2$ for all candidate ports, \eqref{eq:optimal_port} approaches classical SIR selection $\arg\max_k H_k/J_k$. When interference is weak, discarding $\sigma_w^2$ is not generally optimal.
\end{corollary}

\begin{corollary}[Intrinsic finite-sample floor]
For fixed $N_s<\infty$ and $\alpha=P_{D,1}/P_{D,0}>1$, letting $Z_k\to\infty$ gives $\rho_k\to\alpha$ and hence
\begin{equation}
 p_{e,\infty}=p_e(1,\alpha)>0.
\label{eq:intrinsic_floor}
\end{equation}
Thus no amount of channel gain or port diversity removes the finite-$N_s$, finite-$\alpha$ detector floor. For the numerical default $N_s=120$, $\alpha=10$, \eqref{eq:intrinsic_floor} is $1.97\times10^{-34}$ and is therefore invisible on the plotted scales, but the structural floor is retained explicitly. This agrees with the intrinsic NoiseMod floor highlighted in recent FAS-NoiseMod analysis \cite{zayyani_fas_noisemod}.
\end{corollary}

This result is deliberately an oracle statement: it assumes instantaneous knowledge of $H_k$ and $J_k$ at all candidate ports. Section~\ref{sec:acquisition} introduces a finite-sample sensing baseline and Section~\ref{sec:results} quantifies the gap.

\section{Independent-Port Theory With Common Interferer Bits}
\label{sec:iid}
An iid fading-port benchmark is useful analytically, but care is required. Even if $\{H_k,G_{1,k},\ldots,G_{N_I,k}\}$ are independent across $k$, the transmitted bit vector $\bm B$ in \eqref{eq:J} is common to every port. Hence the $Z_k$ are independent \emph{conditional on} $\bm B$, not unconditionally.

For a fixed bit vector $\bm b$, let
\begin{equation}
 F_{Z|\bm b}(z)=\E_{J|\bm b}\left[F_H\!\left(z(\sigma_w^2+J)\right)\right].
\label{eq:Fzcond}
\end{equation}
Conditional iid order statistics then give
\begin{equation}
 F_{Z_{\max}}(z)=2^{-N_I}\sum_{\bm b\in\{0,1\}^{N_I}}\left[F_{Z|\bm b}(z)\right]^{N_p}.
\label{eq:iid_vector}
\end{equation}
For identical interferers, only $M=\sum_iB_i$ matters. Since $M\sim\operatorname{Binomial}(N_I,1/2)$,
\begin{equation}
 F_{Z_{\max}}(z)=\sum_{m=0}^{N_I}\binom{N_I}{m}2^{-N_I}\left[F_{Z|M=m}(z)\right]^{N_p}.
\label{eq:iid_binomial}
\end{equation}
Let $\psi(z)$ denote \eqref{eq:exact_bep} after substituting $V_0\propto1+P_{D,0}z$ and $V_1\propto1+P_{D,1}z$. The exact iid-port average BEP can be written as the Stieltjes integral
\begin{equation}
 \bep=\sum_{m=0}^{N_I}\binom{N_I}{m}2^{-N_I}\int_0^\infty \psi(z)\,d\!\left[F_{Z|M=m}(z)^{N_p}\right].
\label{eq:iid_bep}
\end{equation}
Equations~\eqref{eq:iid_binomial}--\eqref{eq:iid_bep} are exact for iid fading ports under the synchronous random-bit model. In contrast, the tempting expression
\begin{equation}
 \left[\sum_m w_mF_{Z|m}(z)\right]^{N_p}
\end{equation}
mixes over the common bit state \emph{before} taking the port maximum and is generally unequal to $\sum_mw_mF_{Z|m}(z)^{N_p}$.

\section{Correlated Evaluation}
\label{sec:method}
\subsection{Conditional Monte Carlo for Correlated Ports}
For the Jakes-correlated benchmark, a tractable closed-form joint distribution of the ratios $Z_k=H_k/(\sigma_w^2+J_k)$ is not assumed. We instead draw channel/interferer states from \eqref{eq:kmu_joint} and \eqref{eq:J}, choose a port under each selection rule, and evaluate the exact conditional BEP \eqref{eq:exact_bep}. For $L$ channel states, the estimator is
\begin{equation}
 \widehat{\bep}=\frac1L\sum_{\ell=1}^{L}p_e\!\left(V_0^{(\ell)},V_1^{(\ell)}\right).
\label{eq:rbmc}
\end{equation}
This Rao--Blackwellized/conditional Monte Carlo estimator integrates out the $N_s$ received-sample randomness exactly and therefore resolves small BEPs far more efficiently than error counting. Reported uncertainty is the standard error of the terms in \eqref{eq:rbmc}. A separate received-sample simulation validates the detector implementation.

The baseline port rules are: (i) the exact noise-aware oracle \eqref{eq:optimal_port}; (ii) SIR-FAMA, $\arg\max H_k/J_k$; (iii) max-$H$ FAS; (iv) min-$J$; and (v) a fixed port. A random port has the same ensemble-average marginal performance as a fixed port under the symmetric stationary port model and is therefore not plotted separately.

\subsection{Variance-Domain Partial Probing Under Frozen States}
\label{sec:acquisition}
To expose the sensing difficulty behind the oracle, choose $M\le N_p$ uniformly spaced probed ports. We deliberately make the favorable assumption that channels and interferer bit states remain constant throughout all sensing samples and the following data bit. Because the sensing interval can span many nominal $N_s$-sample bit durations, this extends the interferer-state coherence beyond the one-bit synchronous model of Section~\ref{sec:model}; it is therefore a diagnostic lower bound on sensing difficulty, not an implementable fast-FAMA protocol under bit-to-bit-changing interference. At a probed port $k$, collect $L_s$ desired-silent samples and $L_s$ samples with a known Gaussian pilot of variance $P_p$. Their sample-mean energies satisfy
\begin{align}
 A_k&\sim\operatorname{Gamma}\!\left(L_s,\frac{C_k}{L_s}\right),\label{eq:A}\\
 B_k&\sim\operatorname{Gamma}\!\left(L_s,\frac{C_k+H_kP_p}{L_s}\right).
\label{eq:B}
\end{align}
Use
\begin{equation}
 \widehat C_k=A_k,\qquad \widehat H_k=\left[\frac{B_k-A_k}{P_p}\right]^+,
\end{equation}
and select
\begin{equation}
 \widehat k=\arg\max_{k\in\mathcal S_M}\frac{\widehat H_k}{\widehat C_k}.
\label{eq:practical_select}
\end{equation}
This is not a channel-reconstruction algorithm; it is a direct variance-domain estimator of the theorem's metric over a subset.

To include threshold uncertainty, the selected port is optionally calibrated with $L_c$ known low-state samples and $L_c$ known high-state samples. Their unconstrained sample-mean estimates are ordered through the two-variance constrained ML fit. When $\widehat V_1^{\rm raw}\le\widehat V_0^{\rm raw}$, equal calibration sample counts place the constrained optimum on the boundary at the pooled mean $\widehat V_0=\widehat V_1=(\widehat V_0^{\rm raw}+\widehat V_1^{\rm raw})/2$; the threshold uses the continuous ordered-model limit $\eta\to\widehat V_0$. BEP is always evaluated under the true $V_0,V_1$, and the frequency of such boundary events is archived so that the calibration pathology is visible rather than silently clipped. The acquisition cost for one following data bit is
\begin{equation}
 N_{\rm acq}=2ML_s+2L_c,\qquad \xi=\frac{N_s}{N_s+N_{\rm acq}},
\label{eq:overhead}
\end{equation}
where $\xi$ is the data-sample fraction under this deliberately stringent one-bit accounting. If a state could be reused over many bits, the overhead could be amortized; however, changing interferer bits also change $J_k$, so such amortization is not automatically valid for fast FAMA.

\section{Numerical Results}
\label{sec:results}
Unless stated otherwise, $N_s=120$, $\alpha=10$, $\sigma_w^2=1$, $\Omega=1$, $\kappa=1.5$, and $\mu=2$. The desired and each interference waveform use \eqref{eq:levels}.

\subsection{Why the Correct Modeling Choices Matter}
Figure~\ref{fig:ablation}(a) isolates the detector, interference weighting, and port-selection conventions at $N_p=8$, $W=1$, and $N_I=15$. At a desired average variance of $10$ dB, the CLT/harmonic detector predicts $2.82\times10^{-4}$ instead of the exact $2.43\times10^{-4}$, a $16.1\%$ overestimate. Replacing random two-level interference by equal per-user power predicts $1.67\times10^{-4}$, a $31.3\%$ underestimate. SIR and noise-aware selection nearly coincide here because $15$ interferers make the system strongly interference dominated.

Figure~\ref{fig:ablation}(b) shows why that last observation must not be generalized. With $N_I=4$ and weak interference ($-15$ dB per interferer), the exact noise-aware rule achieves $8.12\times10^{-7}$ while SIR-FAMA gives $2.44\times10^{-5}$; max-$H$ FAS is close to optimal because thermal noise dominates. At $0$ and $5$ dB per interferer, SIR approaches the noise-aware rule, consistent with the corollary to Theorem~\ref{thm:selection}. The plot therefore separates the noise-limited and interference-limited roles of FAS and FAMA rather than treating SIR selection as universally exact.

\begin{figure*}[t]
\centering
\includegraphics[width=0.98\textwidth]{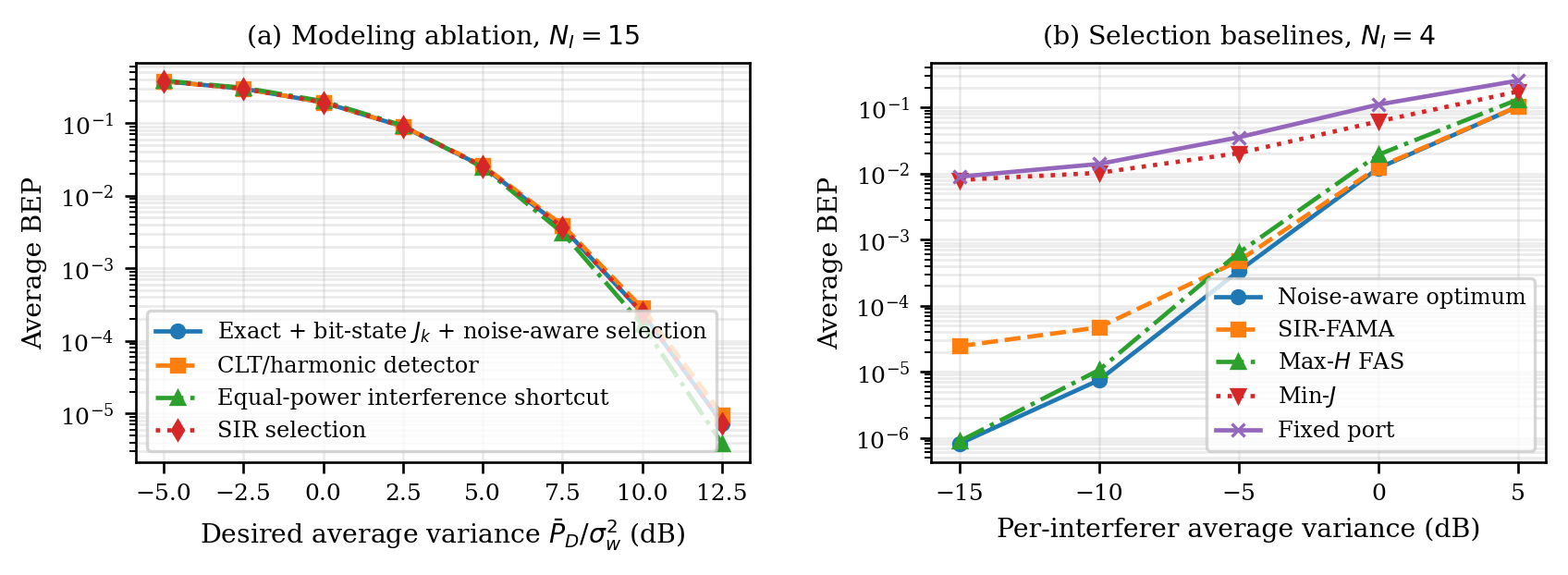}
\caption{(a) Ablation of the exact finite-$N_s$ detector, random-bit interference weighting, and selection convention. (b) Port-selection baselines versus per-interferer average variance, showing the large SIR penalty when thermal noise is non-negligible and convergence toward SIR-FAMA as interference dominates.}
\label{fig:ablation}
\end{figure*}

\subsection{Load Scaling and Admission Capacity}
A multiple-access paper should measure supported co-channel load, not only BEP at one fixed $N_I$. The nominal metric is
\begin{equation}
 N_I^{\max}(\epsilon)=\max\{N_I:\widehat{\bep}(N_I)\le\epsilon\}.
\label{eq:loadmetric}
\end{equation}
Because a discrete admission boundary can be sensitive to Monte Carlo uncertainty, the reported panel uses the more conservative pointwise rule
\begin{equation}
 N_{I,95}^{\max}(\epsilon)=\max\{N_I:\widehat{\bep}(N_I)+1.96\,\mathrm{SE}(N_I)\le\epsilon\}.
\label{eq:loadmetric95}
\end{equation}
The factor $1.96$ is the upper endpoint of the usual pointwise normal-approximation 95\% interval; it is a screening convention, not a simultaneous confidence guarantee over all tested loads. Figure~\ref{fig:load} uses desired average variance $5$ dB, per-interferer average variance $0$ dB, and $\epsilon=10^{-2}$. At $W=2$, increasing $N_p$ shifts the BEP-versus-load curve substantially before the fixed-aperture benefit levels off. Using the conservative rule \eqref{eq:loadmetric95}, the single-port receiver supports only $N_{I,95}^{\max}=2$. At $W=4$, $N_p=16$ supports $N_{I,95}^{\max}=16$, an eightfold increase in admissible co-channel users for the stated target and power normalization. The uncertainty-aware rule matters at some boundaries: for example, at $W=2$, the nominal values for $N_p=2$ and $4$ are six and eleven, while the conservative values are five and ten. At smaller apertures the useful port count levels off earlier: for $W=0.5$, all simulated $N_p\ge2$ cases conservatively support seven interferers. The detailed $W$ dependence need not be monotone at every discrete $(W,N_p)$ point because the Jakes correlation itself oscillates with separation; the robust trend is that additional \emph{spatial degrees of freedom}, not port count alone, determine the load gain.

\begin{figure*}[t]
\centering
\includegraphics[width=0.98\textwidth]{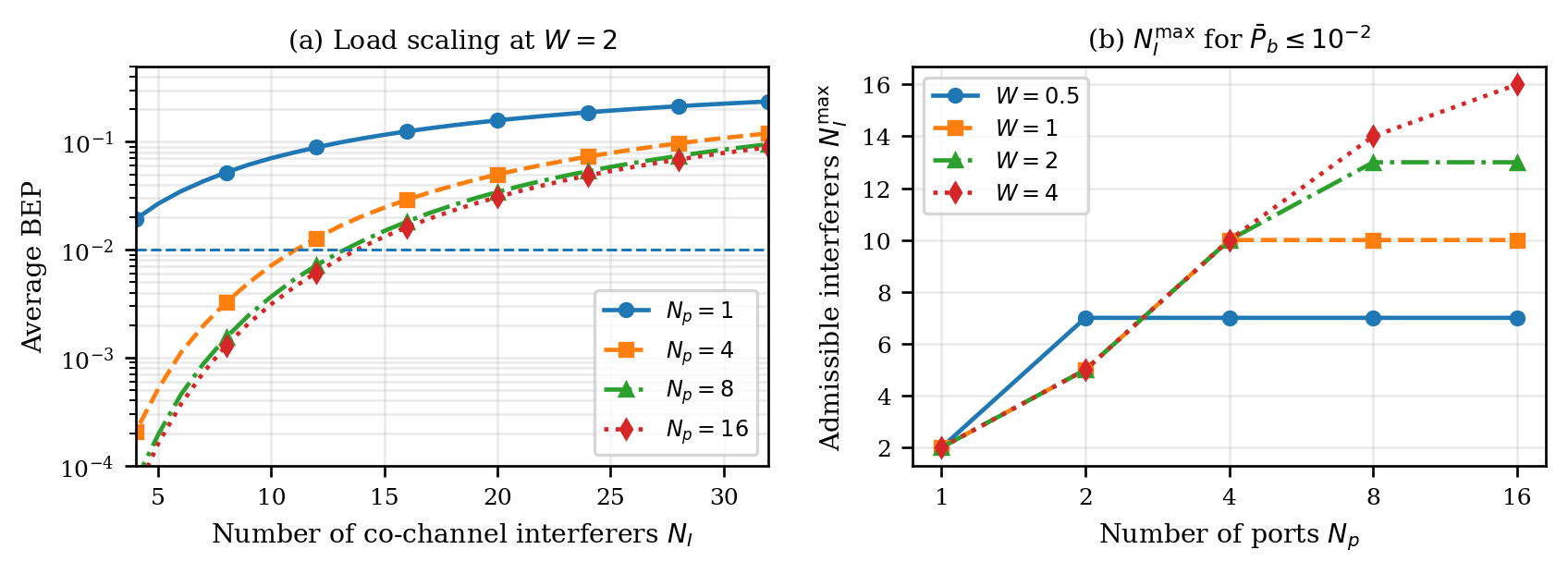}
\caption{Multiple-access load scaling. Desired average variance is $5$ dB, each interferer has average variance $0$ dB, and the target is $10^{-2}$. Panel (b) reports the conservative rule in \eqref{eq:loadmetric95}, not merely thresholding the Monte Carlo point estimate.}
\label{fig:load}
\end{figure*}

\begin{figure*}[t]
\centering
\includegraphics[width=0.98\textwidth]{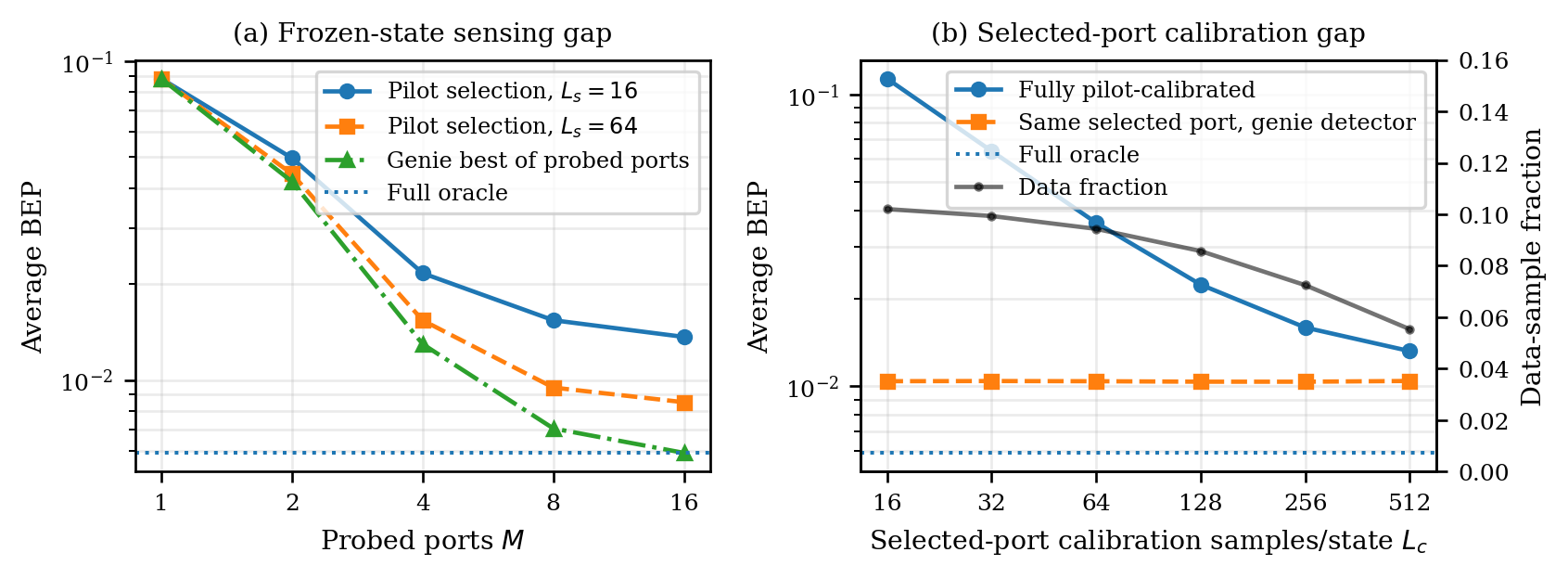}
\caption{Finite-sample variance-domain sensing diagnostic at $N_p=16$, $W=2$, $N_I=12$ under frozen channels and interferer states. (a) Partial probing isolates the port-selection gap by using a genie detector after the selected port. (b) Plug-in threshold calibration exposes additional estimation error and the corresponding one-bit data-sample fraction. The frozen-state assumption is deliberately favorable.}
\label{fig:practical}
\end{figure*}

\subsection{Independent-Port Theory and the Common-Bit Coupling}
Figure~\ref{fig:iid} validates \eqref{eq:iid_bep} for $N_I=8$ and desired average variance $5$ dB. The bit-count-conditioned order-statistic mixture agrees with direct iid multiport simulation to within $1.1\%$ at every plotted $N_p$ and within one Monte Carlo standard error in the archived check. The naive unconditional $F_Z^{N_p}$ rule becomes increasingly optimistic because it effectively redraws the interferer state independently at every port. At $N_p=16$, it predicts $4.71\times10^{-5}$ instead of the correct mixture value $2.75\times10^{-4}$, an optimism factor of $5.83$. Thus, independence of spatial fading does not imply unconditional independence of the complete per-port decision metric when the same random transmitted bits illuminate all ports.

\begin{figure}[t]
\centering
\includegraphics[width=\columnwidth]{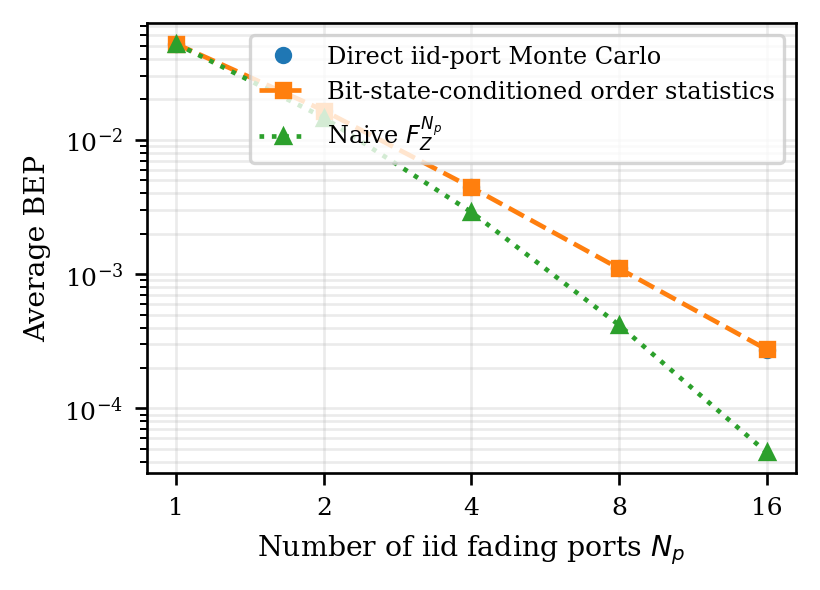}
\caption{Independent fading ports with the same synchronous interferer bits at all ports. Conditioning on the common bit count yields the correct order statistic; mixing first and using an unconditional $F_Z^{N_p}$ is overly optimistic.}
\label{fig:iid}
\end{figure}

\subsection{Near--Far Robustness}
Figure~\ref{fig:nearfar} keeps the total average co-channel variance fixed with $N_I=8$ while increasing the ratio between one dominant interferer and each of the remaining weak interferers. The noise-aware and SIR rules remain near $1.2$--$1.8\times10^{-3}$ across the sweep, whereas max-$H$ FAS degrades from $2.14\times10^{-3}$ to $9.92\times10^{-3}$. The fixed port remains around $5\times10^{-2}$. The result shows that desired-signal diversity alone does not reproduce the interference-avoidance benefit when co-channel power becomes spatially dominated by one user.

\begin{figure}[t]
\centering
\includegraphics[width=\columnwidth]{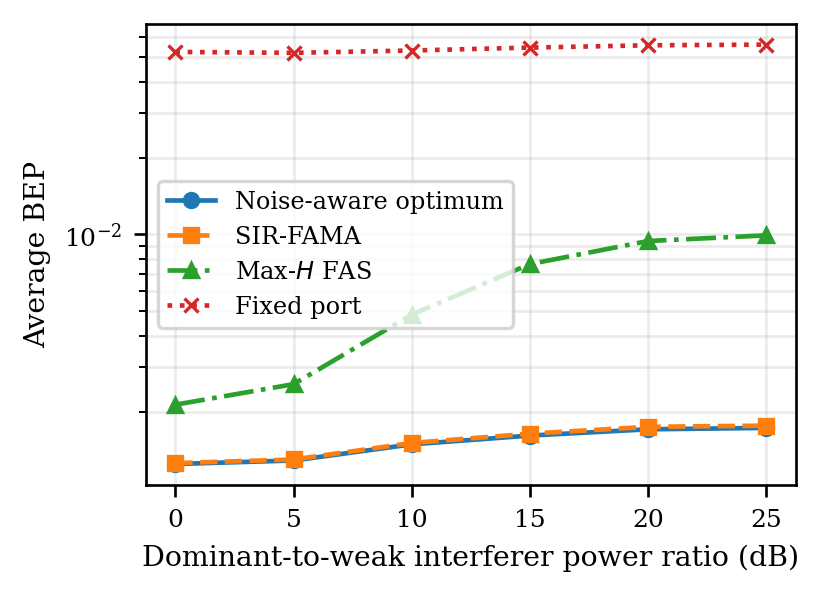}
\caption{Near--far stress test with fixed total average co-channel variance. One interferer is strengthened while the other seven are weakened so the total mean remains constant.}
\label{fig:nearfar}
\end{figure}

\subsection{The Oracle-to-Finite-Sample-Sensing Gap}
Figure~\ref{fig:practical} tests the acquisition bridge of Section~\ref{sec:acquisition} at $N_p=16$, $W=2$, $N_I=12$, desired average variance $5$ dB, and a $20$ dB selection pilot. The full oracle BEP is $5.90\times10^{-3}$. With a genie detector after the selected port, probing more ports improves selection: for $L_s=64$, the BEP falls from $8.82\times10^{-2}$ at one probed port to $8.51\times10^{-3}$ at all 16 ports. However, all-port probing at this sample count consumes $2048$ acquisition samples before detector calibration, leaving only $\xi\approx5.5\%$ data samples under the one-bit accounting of \eqref{eq:overhead}.

The second panel fixes $M=16$, $L_s=32$ and adds threshold calibration. The ordered-fit boundary event $\widehat V_1^{\rm raw}\le\widehat V_0^{\rm raw}$ occurs in $8.64\%$ of trials at $L_c=16$, $0.923\%$ at $L_c=64$, and only $0.0017\%$ at $L_c=512$. With $L_c=512$ samples per desired state, the fully plug-in receiver reaches $1.32\times10^{-2}$, still above the same-selected-port genie-detector value near $1.04\times10^{-2}$ and well above the full oracle. The data fraction is again only about $5.5\%$. This is intentionally a negative sensing result: even under the favorable frozen-state assumption, direct per-port variance acquisition approaches the oracle only by spending a large number of samples. It supports, rather than contradicts, the recent fast-FAMA literature's emphasis on partial observation, reconstruction, and selection-oriented sensing \cite{waqar_fast_fama,elganimi_channel_est}.

\section{Discussion, Scope, and Implications}
\label{sec:discussion}
\subsection{Novelty Boundary}
Exact finite-sample Gamma/ML detection is prior art \cite{alshawaqfeh_tnm}, and SINR-aware FAMA selection already appears in noisy/slow-FAMA work \cite{yang_slow_noisy,eskandari_cgan}. Here, the new analytical point is the exact NoiseMod-specific equivalence between finite-sample BEP minimization and the \emph{current bit-state-aware} statistic in \eqref{eq:optimal_port}, together with the common-state dependence in \eqref{eq:iid_vector}. Related studies already cover FAS diversity for NoiseMod \cite{zayyani_fas_noisemod}, FAS-assisted OODN over $\kappa$-$\mu$ fading \cite{araujo_oodn_fas}, random co-channel OODN interference \cite{araujo_oodn_interference}, and noise-domain NOMA using distinct statistical dimensions \cite{yapici_ndnoma}.

OODN and the present two-level model are not interchangeable: OODN has a nominally off state, whereas both $P_0$ and $P_1$ here are nonzero. Conversely, our synchronous block assumption excludes interferer-state changes inside the $N_s$-sample decision window. With such changes, the conditional samples need not be identically distributed and the optimal likelihood generally becomes a weighted-energy problem rather than \eqref{eq:gamma_stat}.

\subsection{Oracle, Sensing, and Channel Scope}
Theorem~\ref{thm:selection} is an upper benchmark, not a claim of instantaneous knowledge of $H_k$ and $J_k$ at every port. Full-port signal/interference knowledge and selection-oriented reconstruction are recognized FAMA bottlenecks \cite{waqar_fast_fama,elganimi_channel_est}. Our probing experiment is even more favorable than symbol-by-symbol fast FAMA because it freezes interferer states throughout sensing; its poor sample efficiency therefore understates the full acquisition challenge.

The correlated $\kappa$-$\mu$ process in \eqref{eq:kmu_joint} is likewise a benchmark: it uses Jakes-correlated scattered cluster fields and one dominant-component phase over the aperture. User-specific angles, dominant-phase progression, nonisotropic scattering, noninteger $\mu$, or unequal desired/interference covariance models can change the numerical curves. The detector and selection theorem themselves require only realized nonnegative $H_k,J_k$ and the conditional Gaussian variance model.

\subsection{Interpretation of the Load Metric}
Neither \eqref{eq:loadmetric} nor \eqref{eq:loadmetric95} is a universal network capacity. Both depend on the stated powers, BEP target, synchronous activity, ideal switching, and oracle state knowledge; \eqref{eq:loadmetric95} is only a pointwise Monte Carlo screening rule. The metric is nevertheless useful because it converts FAMA gain into an explicit co-channel admission question. A system-level capacity study would additionally require traffic, coding, user geometry, path loss, and net rate after sensing overhead.

\section{Conclusion}
We analyzed FAMA for two-level NoiseMod with random co-channel NoiseMod bits. Exact finite-sample detection implies that conditional BEP decreases strictly with the variance ratio, yielding the BEP-optimal oracle metric $H_k/(\sigma_w^2+J_k)$. Random bit-state scaling gives an exact aggregate-interference transform and moments but generally prevents a single common-scale $\kappa$-$\mu$ collapse. Even with iid fading ports, shared interferer bits create common-state dependence, so conditioning must precede the port order statistic. Under the stated Jakes-correlated benchmark, FAMA increases the admissible co-channel load and remains useful under near--far imbalance. The finite-sample sensing experiment, despite deliberately frozen interference states, still incurs large overhead; practical fast NoiseMod-FAMA therefore remains fundamentally a selection-oriented sensing problem.

\end{document}